\documentclass[12pt,reqno]{amsart} 

\usepackage[margin=1in]{geometry}
\usepackage[foot]{amsaddr}
\usepackage[utf8]{inputenc}
\usepackage[shortlabels]{enumitem}

\usepackage{
  mathtools,
  booktabs,
  braket,
  caption,
  epsfig,
  etoolbox,
  float,
  hyperref,
  latexsym,
  lipsum,
  mathrsfs,
  setspace,
  subcaption,
  url,
  graphicx,
  xcolor
}

\usepackage{caption, subcaption}

\makeatletter
\def\@setthanks{\vspace{-\baselineskip}\def\thanks##1{\@par##1\@addpunct.}\thankses}
\makeatother

\usepackage{lscape}
\newcommand{\bland}{\begin{landscape}}
\newcommand{\eland}{\end{landscape}}

\newcommand{\bburl}[1]{\textcolor{blue}{\url{#1}}}

\definecolor{maroon}{rgb}{0.5, 0.0, 0.0}
\hypersetup{breaklinks=true,
            bookmarks=true,
            pdfauthor={Apoorva Lal},
             pdfkeywords = {},
            colorlinks=true,
            citecolor=maroon,
            urlcolor=blue,
            linkcolor=blue,
            pdfborder={0 0 0}}
\makeatletter
\def\maxwidth{\ifdim\Gin@nat@width>\linewidth\linewidth\else\Gin@nat@width\fi}
\def\maxheight{\ifdim\Gin@nat@height>\textheight\textheight\else\Gin@nat@height\fi}
\makeatother
\setkeys{Gin}{width=\maxwidth,height=\maxheight,keepaspectratio}

\newcommand{\burl}[1]{\textcolor{blue}{\url{#1}}}

\numberwithin{equation}{section}

\makeatletter
\def\section{\@startsection{section}{1}%
      \z@{.7\linespacing\@plus\linespacing}{.5\linespacing}%
      {\normalfont\Large\bfseries\centering }}

\def\sectionL{\@startsection{section}{1}%
      \z@{.7\linespacing\@plus\linespacing}{.5\linespacing}%
      {\normalfont\Large\bfseries}}
\makeatother

\patchcmd{\subsection}{\bfseries}{\bfseries\large}{}{}
\patchcmd{\subsubsection}{\itshape}{\bfseries}{}{}

\makeatletter
\def\paragraph{\@startsection{paragraph}{4}%
  \z@\z@{-\fontdimen2\font}%
  {\sffamily \bfseries }}
\makeatother

\usepackage{
  amscd,
  amsfonts,
  amsmath,
  amssymb,
  amsbsy,
  amsthm,
  bm, 
  dsfont,
  cancel, 
  latexsym,
  cases,
  setspace,
  mathrsfs,
  mathtools,
  graphicx,
  placeins,
  xcolor,
  xargs,
  lineno
}

\usepackage{listings}
\definecolor{codegreen}{rgb}{0,0.6,0}
\definecolor{codegray}{rgb}{0.5,0.5,0.5}
\definecolor{codepurple}{rgb}{0.58,0,0.82}
\definecolor{backcolour}{rgb}{0.95,0.95,0.92}
\lstdefinestyle{mystyle}{
    backgroundcolor=\color{backcolour},
    commentstyle=\color{codegreen},
    keywordstyle=\color{magenta},
    numberstyle=\tiny\color{codegray},
    stringstyle=\color{codepurple},
    basicstyle=\ttfamily\footnotesize,
    breakatwhitespace=false,
    breaklines=true,
    captionpos=b,
    keepspaces=true,
    numbers=left,
    numbersep=5pt,
    showspaces=false,
    showstringspaces=false,
    showtabs=false,
    tabsize=2
}
\usepackage[ruled, lined, commentsnumbered, longend]{algorithm2e}
\SetKwInput{kwParam}{Parameter}
\SetKwInput{kwInit}{Param}

\SetCommentSty{mycommfont}

\usepackage{tikz}
\usetikzlibrary{positioning, calc, shapes.geometric, shapes.multipart,
          shapes, arrows.meta, arrows,
          decorations.markings, external, trees}

\tikzset{
    > = stealth,
    every path/.append style = {
        arrows = ->
    },
    hidden/.style = {
        draw = black,
        shape = circle
    }
}

\tikzstyle{Arrow} = [
          thick,
          decoration={
                  markings,
                  mark=at position 1 with {
                  \arrow[thick]{latex}
          }
          },
          shorten >= 3pt, preaction = {decorate}
          ]

\usepackage[colorinlistoftodos,prependcaption,]{todonotes}

\newcommand{\beq}{\begin{equation}}
\newcommand{\eeq}{\end{equation}}

\definecolor{ao}{rgb}{0.0, 0.5, 0.0}
\definecolor{purp}{HTML}{5601A4}
\definecolor{navy}{HTML}{0D3D56}
\definecolor{ruby}{HTML}{9a2515}
\definecolor{corn}{HTML}{107895}
\definecolor{daisy}{HTML}{EBC944}
\definecolor{coral}{HTML}{F26D21}
\definecolor{kelly}{HTML}{829356}
\definecolor{cranb}{HTML}{E64173}
\definecolor{jet}{HTML}{131516}
\definecolor{ash}{HTML}{555F61}
\definecolor{slate}{HTML}{314F4F}

\newcommand{\epsi}{\varepsilon}

\newcommand*\Bigpar[1]{\left( #1 \right )}

\newcommand*\SetB[1]{\left\{ #1 \right\}}

\newcommand{\Ubr}[2]{\underbrace{#1}_{\text{#2}}}
\newcommand{\Obr}[2]{ \overbrace{#1}^{\text{#2}}}

\newcommand{\ba}{\begin{array}}
\newcommand{\ea}{\end{array}}
\newcommand{\be}{\begin{enumerate}}
\newcommand{\ee}{\end{enumerate}}
\newcommand{\bi}{\begin{itemize}}
\newcommand{\ei}{\end{itemize}}
\newcommand{\I}{\item}
\newcommand{\bs}{\begin{align}\begin{split}\nonumber}
\newcommand{\bsnumber}{\begin{align}\begin{split}}
\newcommand{\es}{\end{split}\end{align}}

\newcommand{\yc}{\ensuremath{Y^{0}}}
\newcommandx{\Yw}[2][1=1,2=]{\ensuremath{Y^{(#1)}_{#2}}}

\newcommandx{\deriv}[2][1=x,2=f]{\nabla \, #2 \Bigpar{ #1 } }
\newcommandx{\ortho}[1][1=L]{#1^{\bot}}

\newcommandx*\seqq[3][1=1,2=x, 3=n]{#2_{#1},\ldots,#2_{#3}}
\newcommandx*\coord[3][1=1,2=x, 3=n]{(#2_{#1},\ldots,#2_{#3})}

\newcommand*\Indic[1]{\mathds{1}\{#1\}}

\renewcommand{\to}{{\rightarrow}}

\newcommand{\R}{\ensuremath{\mathbb{R}}}

\newcommand\frakfamily{\usefont{U}{yfrak}{m}{n}}
\DeclareTextFontCommand{\textfrak}{\frakfamily}

\newcommand{\ooN}{\frac{1}{n}}  

\newcommand{\defeq}{\vcentcolon=}
\newcommand{\eqdef}{=\vcentcolon}

\renewcommand{\to}{{\rightarrow}}

\def\mbi#1{\boldsymbol{#1}} 
\def\ve#1{\mbi{#1}} 
\renewcommand{\vee}[1]{\mathbf{#1}} 

\newcommand{\wh}[1]{\widehat{#1}} 
\newcommand{\wt}[1]{\widetilde{#1}} 

\renewcommand{\iff}{\Leftrightarrow}

\newcommand{\E}{\mathbb{E}} 
\newcommand{\Vv}{\mathbb{V}} 

\newcommand\indep{\protect\mathpalette{\protect\independenT}{\perp}}
\def\independenT#1#2{\mathrel{\rlap{$#1#2$}\mkern5mu{#1#2}}}

\newcommand{\Exp}[1]{\mathbb{E}\left[#1\right]}

\newcommand{\Var}[1]{\mathbb{V}\left[#1\right]}
\newcommand{\Covar}[1]{\text{Cov}\left[#1\right]}

\newcommand*\Unif[1]{\mathsf{U} \left[ #1 \right ]}

\newcommand{\hyp}[2]{
\ensuremath{H_0:#1 \ifhmode\quad\text{versus}\quad\fi\text{ vs. } H_1:#2}}

\newcommandx{\uniff}[1][1={a,b}]{\textrm{Unif}\left({#1}\right)}
\newcommandx{\unifd}[1][1={a,\ldots,b}]{\textrm{Unif}\left\{{#1}\right\}}
\newcommandx{\dunif}[3][1=x,2=a,3=b]{\frac{I(#2<#1<#3)}{#3-#2}}
\newcommandx{\dunifd}[3][1=x,2=a,3=b]{\frac{I(#2\le#1\le#3)}{#3-#2+1}}
\newcommandx{\punif}[3][1=x,2=a,3=b]{
\begin{cases} 0 & #1 < #2 \\ \frac{#1-#2}{#3-#2} & #2 < #1 < #3 \\ 1 & #1 > #3\\\end{cases}}
\newcommandx{\punifd}[3][1=x,2=a,3=b]{
\begin{cases} 0 & #1 < #2\\ \frac{\lfloor#1\rfloor-#2+1}{#3-#2} & #2 \le #1 \le #3 \\ 1 & #1 > #3\\ \end{cases}}

\newcommandx\bern[1][1=p]{\textrm{Bern}\left({#1}\right)}
\newcommandx\dbern[2][1=x,2=p]{#2^{#1} \left(1-#2\right)^{1-#1}}
\newcommandx\pbern[2][1=x,2=p]{\left(1-#2\right)^{1-#1}}

\newcommandx\bin[1][1={n,p}]{\textrm{Bin}\left(#1\right)}
\newcommandx\dbin[3][1=x,2=n,3=p]{\binom{#2}{#1}#3^#1\left(1-#3\right)^{#2-#1}}

\newcommandx\mult[1][1={n,p}]{\textrm{Mult}\left(#1\right)}
\newcommandx\dmult[3][1=x,2=n,3=p]{\frac{#2!}{#1_1!\ldots#1_k!}#3_1^{#1_1}\cdots#3_k^{#1_k}}

\newcommandx\hyper[1][1={N,m,n}]{\textrm{Hyp}\left({#1}\right)}
\newcommandx\dhyper[4][1=x,2=N,3=m,4=n]{\frac{\binom{#3}{#1}\binom{#2-#3}{#4-#1}}{\binom{#2}{#4}}}

\newcommandx\nbin[1][1={r,p}]{\textrm{NBin}\left({#1}\right)}
\newcommandx\dnbin[3][1=x,2=r,3=p]{\binom{#1+#2-1}{#2-1}#3^#2(1-#3)^#1}
\newcommandx\pnbin[3][1=x,2=r,3=p]{I_#3(#2,#1+1)}

\newcommandx\geo[1][1=p]{\textrm{Geo}\left(#1\right)}
\newcommandx\dgeo[2][1=x,2=p]{#2(1-#2)^{#1-1}}
\newcommandx\pgeo[2][1=x,2=p]{1-(1-#2)^#1}

\newcommandx\pois[1][1=\lambda]{\textrm{Po}\left({#1}\right)}
\newcommandx\dpois[2][1=x,2=\lambda]{\frac{#2^#1 e^{-#2}}{#1!}}
\newcommandx\ppois[2][1=x,2=\lambda]{e^{-#2}\sum_{i=0}^#1\frac{#2^i}{i!}}

\newcommandx\normall[1][1={\mu,\sigma^2}]{\mathcal{N}\left({#1}\right)}
\newcommandx\dnormall[3][1=x,2=\mu,3=\sigma]%
  {\frac{1}{#3\sqrt{2\pi}}\exp \Bigpar{-\frac{\left(#1-#2\right)^2}{2 #3^2}}}
\newcommandx\pnormall[1][1=x]{\Phi\left({#1}\right)}
\newcommandx\qnormall[1]{\Phi^{-1}\left({#1}\right)}

\newcommandx\mvn[1][1={\mu,\Sigma}]{\mathrm{MVN}\left({#1}\right)}

\newcommandx\ex[1][1=\lambda]{\textrm{Exp}\left(#1\right)}
\newcommandx\dex[2][1=x,2=\lambda]{#2e^{-#1 #2}}
\newcommandx\pex[2][1=x,2=\lambda]{1-e^{-#1 #2}}

\newcommandx\gam[1][1={\alpha,\lambda}]{\textrm{Gamma}\left({#1}\right)}
\newcommandx\dgamma[3][1=x,2=\alpha,3=\lambda]%
  {\frac{#3^{#2}}{\Gamma\left( #2 \right)} #1^{#2-1}e^{-#3#1}}

\newcommandx\invgamma[1][1={\alpha,\beta}]{\textrm{InvGamma}\left({#1}\right)}
\newcommandx\dinvgamma[3][1=x,2=\alpha,3=\beta]%
{\frac{#3^{#2}}{\Gamma\left(#2\right)}#1^{-#2-1}e^{-#3/#1}}
\newcommandx\pinvgamma[3][1=x,2=\alpha,3=\beta]%
{\frac{\Gamma\left(#2,\frac{#3}{#1}\right)}{\Gamma\left(#2\right)}}

\newcommandx\bet[1][1={\alpha,\beta}]{\textrm{Beta}\left(#1\right)}
\newcommandx\dbeta[3][1=x,2=\alpha,3=\beta]
{\frac{\Gamma\left(#2+#3\right)}{\Gamma\left(#2\right)\Gamma\left(#3\right)}#1^{#2-1}\left(1-#1\right)^{#3-1}}

\newcommandx\dir[1][1={\alpha}]{\textrm{Dir}\left(#1\right)}
\newcommandx\ddir[3][1=x,2=\alpha]{\frac{\Gamma\left(\sum_{i=1}^k #2_i\right)}{\prod_{i=1}^k\Gamma\left(#2_i\right)}\prod_{i=1}^k #1_i^{#2_i-1}}

\newcommandx\weibull[1][1={\alpha}]{\textrm{Dir}\left(#1\right)}
\newcommandx\dweibull[3][1=x,2=\lambda,3=k]{\frac{#3}{#2}
\left(\frac{#1}{#2}\right)^{#3-1} e^{-(#1/#2)^k}}

\newcommandx\chisq[1][1=k]{\chi_{#1}^2}

\newcommandx\zet[1][1=s]{\textrm{Zeta}\left(#1\right)}
\newcommandx\dzeta[2][1=x,2=s]{\frac{#1^{-#2}}{\zeta\left(#2\right)}}

\makeatletter%
  \@ifclassloaded{amsart}
  { 
    \newtheoremstyle{mystyle}
      {}
      {}
      {}
      {}
      {\sffamily \bfseries }
      {.}
      {\newline }
      {\thmname{#1}\thmnumber{ #2}\thmnote{ (#3)}}
    \theoremstyle{mystyle}
    \newtheorem{thm}{Theorem}[section]
    
    \newtheorem{defi}[thm]{Defn}

    \newtheorem{prop}[thm]{Proposition}

    \renewenvironment{proof}{\noindent{\bf Proof}\hspace*{1em}}{\qed\bigskip\\}
    \newenvironment{proof-sketch}{\noindent{\bf Sketch of Proof}
      \hspace*{1em}}{\qed\bigskip\\}
    \newenvironment{proof-idea}{\noindent{\bf Proof Idea}
      \hspace*{1em}}{\qed\bigskip\\}
    \newenvironment{proof-of-lemma}[1][{}]{\noindent{\bf Proof of Lemma {#1}}
      \hspace*{1em}}{\qed\bigskip\\}
    \newenvironment{proof-of-proposition}[1][{}]{\noindent{\bf
        Proof of Proposition {#1}}
      \hspace*{1em}}{\qed\bigskip\\}
    \newenvironment{proof-of-theorem}[1][{}]{\noindent{\bf Proof of Theorem {#1}}
      \hspace*{1em}}{\qed\bigskip\\}
    \newenvironment{inner-proof}{\noindent{\bf Proof}\hspace{1em}}{
      $\bigtriangledown$\medskip\\}
  }
{}
\makeatother%

\makeatletter
\providecommand\@dotsep{5}
\def\listtodoname{List of Todos}
\def\listoftodos{\@starttoc{tdo}\listtodoname}
\makeatother

\usepackage[
  hyperref=true,
  url=false,
  isbn=false,
  backref=true,
  maxcitenames=3,
  backend=biber,
  style=authoryear,
  citestyle=authoryear,
  natbib=true]{biblatex}

\title[Ranking Reversals in Linear Regression]{Does Regression Produce Representative Causal Rankings?}
\author{Apoorva Lal}
\address{Netflix}
\date{\today}

\begin{document}
\maketitle
\begin{abstract}

We examine the challenges in ranking multiple treatments
based on their estimated effects when using linear regression or its
popular double-machine-learning variant, the Partially Linear
Model (PLM), in the presence of treatment effect heterogeneity. We
demonstrate by example that overlap-weighting performed by linear
models like PLM can produce Weighted Average Treatment Effects (WATE)
that have rankings that are inconsistent with the rankings of the
underlying Average Treatment Effects (ATE). We define this as ranking
reversals and derive a necessary and sufficient condition for ranking
reversals under the PLM. We conclude with several simulation studies
conditions under which ranking reversals occur.
\end{abstract}

\section{Introduction}

In both the public and private sector, ranking treatments based on
their causal effects is crucial for decision-making.  In commercial
applications, it is common to rank user actions by estimating their
effect on a target metric, and subsequently seeking to encourage
actions with large estimated effects, which are deemed `high value'.
An increasingly popular approach is to use Partially Linear Models
(PLM) to flexibly condition on a large set of confounders as part of
estimating causal effects of treatments while relaxing the stringent
form assumptions
\parencite{Chernozhukov2018-fl}. This estimator is rooted in the
seminal Frisch-Waugh-Lovell theorem and is extremely popular in
practice, and is viewed as \emph{the} Double Machine Learning (DML)
estimator by applied users\footnote{This is not strictly correct,
since DML is in fact a recipe for constructing Neyman-orthogonal
estimators for a wide variety of causal and structural parameters.
However, due to the prevalence of conditional-ignorability-based
identification assumptions and the popularity of linear regression,
the PLM has become synonymous with DML. \textcite{Chernozhukov2022-se}
study Neyman-orthogonal estimators for a wide variety of causal and
structural parameters.}.

However, under treatment effect heterogeneity, it is well known from
that linear regression performs overlap-weighting. As a result, it is is biased for the Average Treatment Effect (ATE), but instead estimates a
conditional-variance Weighted average of treatment effects (WATE).
So, when unbiased estimation of treatment effects is the goal, practitioners opt for direct estimation methods such as IPW (Inverse Propensity Weighting) or its Augmented variety (AIPW), or regression imputation / g-modelling. However, in many cases, practioners seek to rank treatment effects instead, and performance of common estimators for ranking purposes is less well-understood.  We first construct an example with two treatments where the ranking of Weighted Average Treatment Effects (WATEs) produced by the PLM is the opposite of the true ranking of underlying Average treatment Effects (ATEs), which we formalize as a `ranking reversal' property that is undesirable for downstream decision-making.  This implies that decision-makers that seek to rank treatments based on the treatment effects may therefore form incorrect rankings if they use PLM coefficients to form these rankings.
We then derive a decomposition relating the WATE and ATE, which gives rise to a necessary and sufficient condition for ranking
reversals, and provide economic intuition for it. We
find that ranking reversals require substantial treatment effect
heterogeneity and covariances between regression weights and treatment
effects to be of opposite signs across the treatments being ranked. We conclude with an array of simulation designs that mimic realistic DGPs that comport with our theoretical findings about the likelihood of rank reversals under different heterogeneity patterns.

\section{A Simple Numerical Example}

Consider a binary covariate $x \sim \text{Bernoulli}(0.5)$ and two
binary treatments $W_1, W_2$ with the following propensity scores:

\begin{table}[h]
\centering
\begin{tabular}{c|cc}
 & $W_1 = 0$ & $W_2 = 1$ \\
\hline
$X = 0$ & 0.01 & 0.5 \\
$X = 1$ & 0.5 & 0.01
\end{tabular}
\end{table}

The true treatment effects are:

\begin{table}[h]
\centering
\begin{tabular}{c|cc}
 & $\tau_1$ & $\tau_2$ \\
\hline
$X = 0$ & -3 & -2 \\
$X = 1$ & 3 & 3 \\
\hline
ATE & 0 & 0.5
\end{tabular}
\end{table}

With linear propensity scores, we can plug the above two sets of
numbers into \ref{eqn:watedecomp} and \ref{eqn:cvwt} to construct PLM
regression coefficients

\begin{align*}
\tilde{\tau}_1 & = \frac{-3 \cdot 0.01 \cdot 0.99 + 3 \cdot 0.5 \cdot
  0.5}{ 0.01 \cdot 0.99 + 0.5 \cdot 0.5 } = 2.7714 \\
\tilde{\tau}_2 & = \frac{-2 \cdot 0.5 \cdot 0.5 + 3 \cdot 0.01 \cdot 0.99}{0.01 \cdot 0.99 + 0.5 \cdot 0.5} = -1.8095
\end{align*}

\begin{figure}
    \centering
    \includegraphics[width=0.5\linewidth]{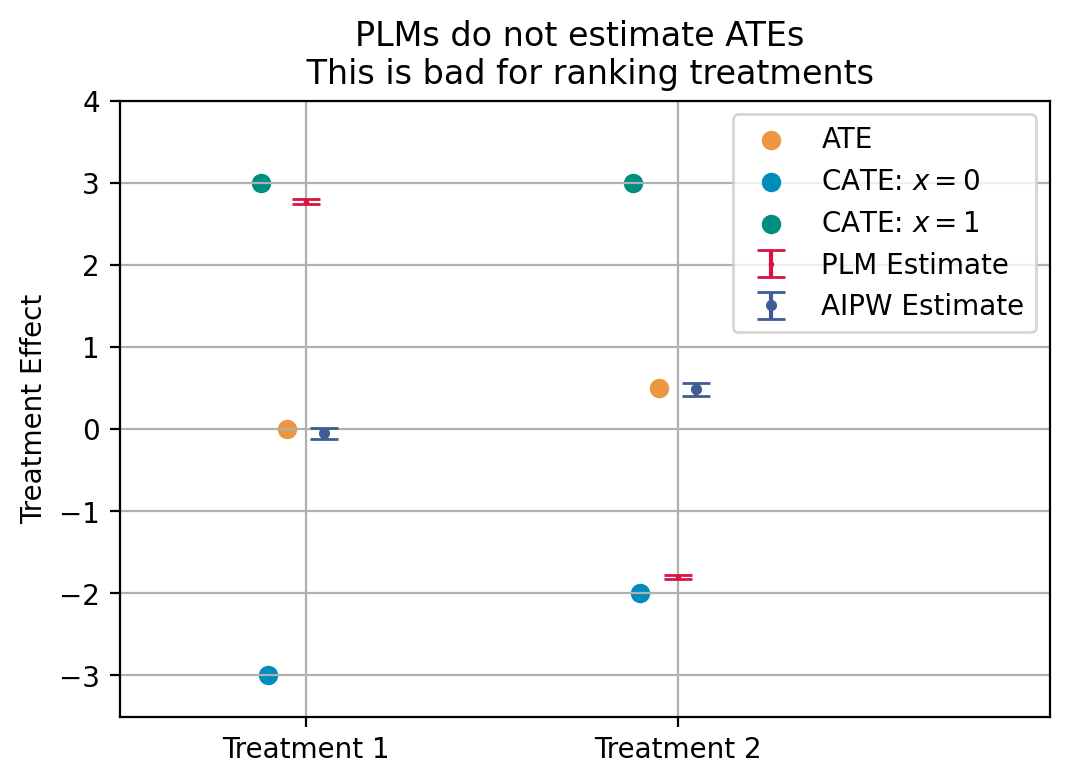}
    \caption{Strata-level and overall true effects, and estimated effects from PLM and AIPW}
    \label{fig:basicsim}
\end{figure}

In contrast, IPW or AIPW correctly recovers the ATEs. These results
demonstrate that PLM leads to incorrect ranking of treatments, while
AIPW provides the correct ranking based on ATEs. This is an admittedly
contrived example; in the next section, we formalize the properties of
this example that yielded the poor ranking performance of PLM.

\section{Methodology}

We consider a setting with multiple binary treatments where for each
unit $i$, we observe an outcome $Y_i \in \mathbb{R}$, treatment
assignment $W_i \in \{1,...,K\}$ indicating which of $K$ treatments
was received (with $W_i = 0$ denoting control), and pre-treatment
covariates $\mathbf{X}_i \in \mathbb{R}^d$. Our goal is to rank
treatments according to their average treatment effects relative to
control, defined as $\tau_j \defeq \mathbb{E}[Y_i(j) - Y_i(0)]$ for
each treatment $j$. We seek to form a poset ordering $(\leq, \ve{\tau}_j)$, and want to estimate $\ve{\tau}_j$s using standard techniques under selection-on-observables assumptions [Unconfoundedness and Overlap \parencite{Imbens2004-ir}].


\begin{defi}[Partially Linear Model]

For each treatment $W_i$, the PLM approach models the outcome as:

\beq
    Y_i = \tau W_i + g(\vee{X}_i) + \varepsilon_i
\eeq

Estimation typically involves a residuals-on-residuals regression:

\beq
    Y_i - \E[Y_i | \vee{X}_i] = \wh{\tau} (W_i - \E[W_i | \vee{X}_i])
    + \eta_i \label{eqn:ror}
\eeq

Where the conditional expectations $\Exp{Y \mid \vee{X}} \eqdef \mu(\vee{X})$ and $\Exp{W
\mid \vee{X}} \eqdef p(\vee{X})$ are estimated using flexible non-parametric regression
methods and cross-fit to avoid over-fitting to satisfy the technical
requirements in Chernozhukov et al (2018).

\end{defi}

\begin{thm}[Conditional Variance weighting property of linear
regression]\label{thm:cvwt}

Under treatment effect heterogeneity, PLM estimates a weighted average treatment effect:

$$
\wh{\tau} = \frac{\Exp{\omega_i \tau_i}}{\Exp{\omega_i}}
$$

where $\omega_i \defeq (W_i - \Exp{W_i \mid X_i})^2$
\parencite{Angrist1998-ok,Angrist1999-sp,Aronow2016-nn}.Defining normalized weights
$\gamma_i = \omega_i/\Exp{\omega_i}$ and working
(without loss of generality) with discrete $\vee{X}$ lets us rewrite
the above as

\beq
    \text{plim } \hat{\tau} = \E[\gamma(\vee{X})\tau(\vee{X})] \eqdef \text{WATE} \label{eqn:watedecomp}
\eeq

where $\gamma(\vee{X})$ are (normalized) weights that depend on the
propensity scores. The weights take the following form

\beq
\gamma(\vee{X}) = \frac{\Vv[D \mid \vee{X}]}{\E[\Vv[D \mid \vee{X}]]}
=
\frac{p(\vee{X}) (1-p(\vee{X}))}{\E[p(\vee{X})(1-p(\vee{X}))]}
\label{eqn:cvwt}
\eeq

where the second equality uses the fact that each treatment is binary
and substitutes in the expression for binomial variance. Proof in
\ref{sec:p1}.
\end{thm}

This means that in the presence of treatment effect heterogeneity
(i.e. $\tau(\vee{X})$ is not a constant function $=\tau$), the
probability limit of the regression coefficient is no longer the
Average Treatment Effect (ATE $\defeq \E[\tau(\vee{X})$) but is
instead the above Weighted Average Treatment Effect (WATE), with
weights $\gamma$ implicitly chosen by the regression specification.
These weights are largest for propensity scores close to 0.5, which
results in OLS performing `overlap-weighting' where it down-weights
strata with extreme propensity scores, and discards strata with no
overlap (with propensity scores equal to 0 or 1).

An interesting alternative but complementary decomposition is studied
by \textcite{SloczynskiUnknown-kg}, who shows that the regression
coefficient $\hat{\tau}$ can also be decomposed into the ATT (Average
Treatment Effect on the Treated) and ATU (Average Treatment Effect on
the Untreated), with weights that are inversely proportional to group
sizes. In other words, the larger the share of the treated group, the
lower weight it receives, and vice versa.

\subsection{Rank Reversal: definition and conditions}

With these weights in hand, we can define the property observed in the
previous section.

\begin{defi}[Rank Reversal]\label{defn:rrev}
For any two treatments $j$ and $k$, a
ranking reversal implies that we have $\text{ATE}_j > \text{ATE}_k$
but $\text{WATE}_j < \text{WATE}_k$. This occurs when

\begin{align}
\Obr{\Exp{\tau_j(\vee{X})}}{ATE$_j$} & >
\Obr{\Exp{\tau_k(\vee{X})}}{ATE$_k$} \label{eqn:atejgk} \\
\Ubr{\Exp{\gamma_j(\vee{X}) \tau_j(\vee{X})}}{WATE$_j$} & < \Ubr{\Exp{\gamma_k(\vee{X})\tau_k(\vee{X})}}{WATE$_k$} \label{eqn:wagejlk}
\end{align}

\end{defi}

We first derive an expression relating the ATE and WATE.  For any
treatment $g$, we can decompose the WATE using the definition of
covariance ($\Covar{a, b} = \Exp{ab} -
\Exp{a}\Exp{b}$)

$$\mathbb{E}[\gamma_g(\mathbf{X})\tau_g(\mathbf{X})] =
\mathbb{E}[\gamma_g(\mathbf{X})]\mathbb{E}[\tau_g(\mathbf{X})] +
\text{Cov}(\tau_g(\mathbf{X}), \gamma_g(\mathbf{X})) 
$$

Note that by construction of regression weights  $\gamma_g(\mathbf{X}) \defeq  \Var{W\mid\vee{X}}/\Exp{\Var{W \mid \vee{X}}}$ have an expected value of 1. So, we arrive at the following decomposition

\beq
\Ubr{\mathbb{E}[\gamma_g(\mathbf{X})\tau_g(\mathbf{X})]}{WATE$_g$}
 =
 \Ubr{\mathbb{E}[\tau_g(\mathbf{X})]}{ATE$_g$} +
 \text{Cov}(\tau_g(\mathbf{X}), \gamma_g(\mathbf{X})) \label{eqn:watedecomp}
\eeq

We provide three simple examples numerically illustrating the above decomposition with negative, zero, and positive covariance between the regression weights and treatment functions in figure~\ref{fig:treatment-effects-sim}. 

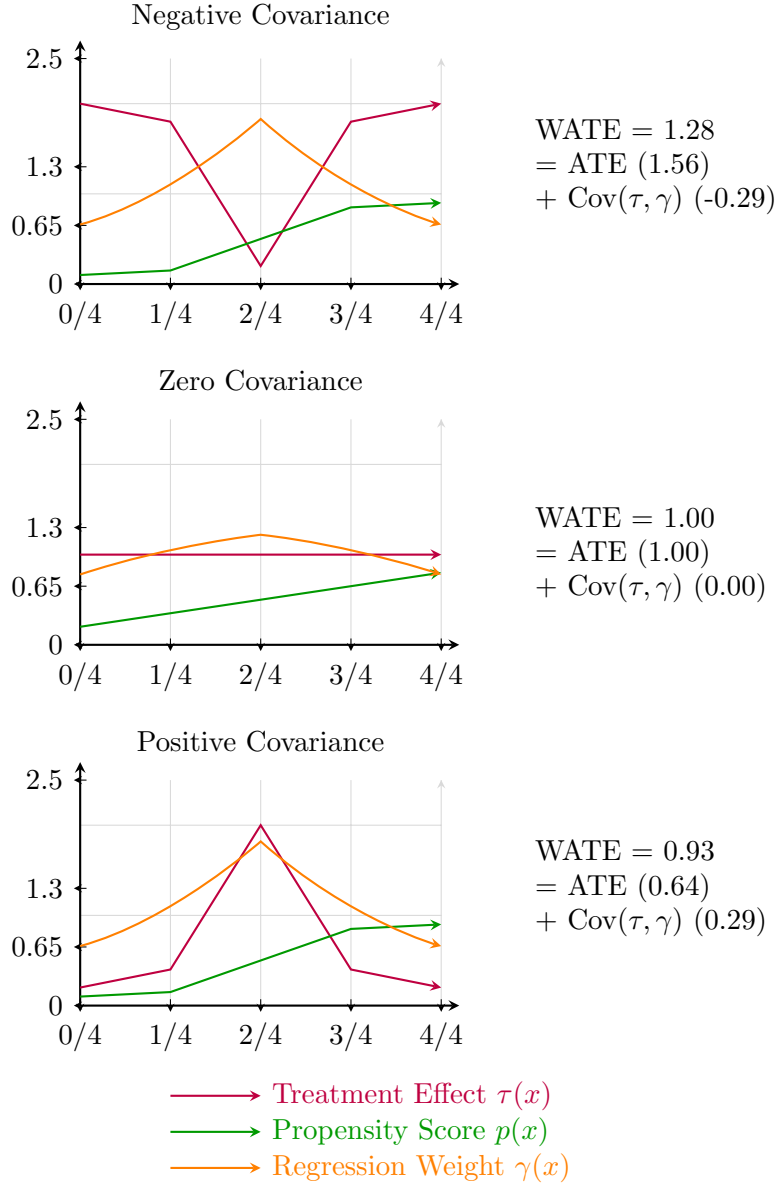
\begin{figure}[t]
    \centering
    \begin{tikzpicture}[scale=1.2]
        \tikzset{
            every node/.style={font=\small},
            estimate/.style={align=left, text width=3.5cm}
        }
        
        \foreach \i in {0,1,2} {
            \begin{scope}[yshift=-4*\i cm]
                \draw[gray!30] (0,0) grid[step=1] (4,2.5);
                \draw[->,thick] (0,0) -- (4.2,0); 
                \draw[->,thick] (0,0) -- (0,2.7) ;
                
                \foreach \x in {0,1,2,3,4} {
                    \draw (\x,2pt) -- (\x,-2pt) node[below] {\x/4};
                }
                \foreach \y in {0,0.65,1.3,2.5} {
                    \draw (2pt,\y) -- (-2pt,\y) node[left] {\y};
                }
            \end{scope}
        }
        
        \node[above] at (2,2.7) {Negative Covariance};
        \node[estimate] at (6.5,1.3) {
            WATE = 1.28 \\
            = ATE (1.56) \\
            + Cov($\tau,\gamma$) (-0.29)
        };
        
        \node[above] at (2,-1.3) {Zero Covariance};
        \node[estimate] at (6.5,-3) {
            WATE = 1.00 \\
            = ATE (1.00) \\
            + Cov($\tau,\gamma$) (0.00)
        };
        
        \node[above] at (2,-5.3) {Positive Covariance};
        \node[estimate] at (6.5,-6.7) {
            WATE = 0.93 \\
            = ATE (0.64) \\
            + Cov($\tau,\gamma$) (0.29)
        };
        
        \draw[purple, thick] (0,2) -- (1,1.8) -- (2,0.2) -- (3,1.8) -- (4,2);
        \draw[green!60!black, thick] (0,0.1) -- (1,0.15) -- (2,0.5) -- (3,0.85) -- (4,0.9);
        \draw[orange, thick] (0,0.66) .. controls (1,0.93) and (2,1.83) .. (2,1.83) .. controls (3,0.93) and (4,0.66) .. (4,0.66);
        
        \begin{scope}[yshift=-4cm]
            \draw[purple, thick] (0,1) -- (4,1);
            \draw[green!60!black, thick] (0,0.2) -- (1,0.35) -- (2,0.5) -- (3,0.65) -- (4,0.8);
            \draw[orange, thick] (0,0.78) .. controls (1,1.11) and (2,1.22) .. (2,1.22) .. controls (3,1.11) and (4,0.78) .. (4,0.78);
        \end{scope}
        
        \begin{scope}[yshift=-8cm]
            \draw[purple, thick] (0,0.2) -- (1,0.4) -- (2,2.0) -- (3,0.4) -- (4,0.2);
            \draw[green!60!black, thick] (0,0.1) -- (1,0.15) -- (2,0.5) -- (3,0.85) -- (4,0.9);
            \draw[orange, thick] (0,0.66) .. controls (1,0.93) and (2,1.82) .. (2,1.82) .. controls (3,0.93) and (4,0.66) .. (4,0.66);
        \end{scope}
        
        \begin{scope}[yshift=-9.0cm]
            \draw[purple, thick] (1,0) -- (2,0) node[right] {Treatment Effect $\tau(x)$};
            \draw[green!60!black, thick] (1,-0.4) -- (2,-0.4) node[right] {Propensity Score $p(x)$};
            \draw[orange, thick] (1,-0.8) -- (2,-0.8) node[right] {Regression Weight $\gamma(x)$};
        \end{scope}
        
    \end{tikzpicture}
    \caption{Treatment effect heterogeneity and regression weights under negative, zero, and positive scenarios for the $\Covar{\tau_g(\vee{X}), \gamma_g(\vee{X}}$ term in \ref{eqn:watedecomp}. We have a single covariate $X$ with $5$ discrete strata with equal probability, and vary propensity scores and treatment effects according to the green and red functions specified above, which gives rise to the orange regression weights function. The right panel for each scenario shows how the weighted average treatment effect (WATE) estimated using regression decomposed into the true average treatment effect (ATE) and the covariance between treatment effects and regression weights.}
    \label{fig:treatment-effects-sim}
\end{figure}

This decomposition immediately illustrates how rank reversals may arise in practice: when the second term in \ref{eqn:watedecomp} is large enough to offset the first, rank-reversals may occur.


\begin{prop}[Necessary and Sufficient Condition for Rank Reversal]

The following condition yields rank-reversal between treatments $j$ and $k$

\beq
\Exp{\tau_j(\vee{X}} + \Covar{\tau_j(\vee{X}, \gamma_j(\vee{X}} < \Exp{\tau_k(\vee{X})} + \Covar{\tau_k(X), \gamma_k(\vee{X}}
\label{eqn:cond}
\eeq

This is an immediate implication of the decomposition \ref{eqn:watedecomp}. Proof in \ref{appdx:directproof}. We also provide slightly more transparent sufficient conditions that parametrises the magnitudes of the two covariances in \ref{eqn:cond} in appdx~\ref{appdx:interpretablesuff}.

\end{prop}

\paragraph{When can we expect PLM coefficients to yield correct
rankings?}

\be
\I Constant treatment effects ($\tau(\vee{X}) = \tau$): Here, PLM,
IPW, and AIPW all estimate the same quantity. This is rare in practice
but serves as a useful benchmark.

\I Uncorrelated weights and effects
($\Covar{\gamma(\vee{X}),\tau(\vee{X})} \approx 0$): This can happen
when:
\be
   \I Treatment assignment is relatively balanced ($p(\vee{X}) \approx 0.5$)
   \I Treatment effects vary independently of variables that predict
   treatment
   \I As-good-as-random assignment: if units don't have the
   opportunity to sort into treatment based on private information
   about their own treatment effects $\tau(\vee{X})$, this covariance
   will be more likely to be small.
\ee
\I Uniform selection on gains: If units sort into treatments $j$ and
$k$ based on private information about their expected gains
$\tau_j(\vee{x}), \tau_k(\vee{x})$, the covariance
$\Covar{\gamma_g(\vee{X}), \tau_g(\vee{X})}$ will be of the same sign
for $g \in \SetB{j,k}$, which would not flip the rankings between the ATEs.
\I Similar propensity score distributions: When $p_j(\vee{X})$ and
$p_k(\vee{X})$ have similar distributions, $\gamma_j(\vee{X})$ and
$\gamma_k(\vee{X})$ will be similar, reducing the chance of rank
reversals. This suggests observational studies with very different
propensity scores across treatments are more prone to rank reversals
\I Moderate treatment effect heterogeneity: If heterogeneity in
treatment effects is modest, and this is known to agents, it is less
likely that they actively seek or avoid treatments (which pushes
$p_g(\vee{X})$ towards 0 or 1) based on this information, which
weakens the magnitude of $\Covar{\tau(\cdot), \gamma(\cdot)}$, which
in turn makes it less likely that the covariances for different
treatments are of contrasting signs to result in rank reversals.
\ee

A practical implication of the above is that when treatment effects
are suspected to be highly heterogeneous with units selecting into
treatments, researchers should prefer AIPW over PLM for ranking.

\begin{defi}[Augmented Inverse-Propensity Weighting (AIPW) Estimators]

An alternative to the PLM that does not fall prey to the ranking
reversal property is the AIPW estimator, which involves construction
of a `pseudo-outcome' $\Gamma_i^{j}$ that is the estimated potential outcome under treatment $j$ \parencite{Cattaneo2010-oc,Chernozhukov2018-fl}

\begin{align*}
\wh{\Gamma}_i^{j} = \wh{\mu}^{j,-k}(\vee{X}_i) +
&
\frac{\Indic{W_i = j}}{\wh{p}^{j,-k_i}(\vee{X}_i)} \Bigpar{Y_i - \wh{\mu}^{j,-k_i}(\vee{X}_i)}  \\
\wh{\tau}^{\text{AIPW}, a, b} = \ooN \sum_{i}^n\Bigpar{\wh{\Gamma}_i^{a} -  \wh{\Gamma}_i^{b} }
\end{align*}

where we first partition data by assigning each observation into $k_i
\in \Unif{K}$ folds, and cross-fit nuisance functions
$\wh{\mu}(\cdot)$ (an outcome regression within treatment level $j$) and $\wh{p}$ (a multi-class propensity score that models the probability of treatment level $j$) so that their predictions for unit $i$ are produced from models that were not trained on the $k_i-$th
fold. The above estimator is consistent for the ATE regardless of the
level of heterogeneity in the underlying treatment effect function
$\tau(\vee{X})$, which implies that it does not exhibit rank-reversal
properties, but conversely may have poor empirical performance in the
presence of extreme propensity scores.

\end{defi}

\section{Numerical Experiments}

\subsection{Simulation Design}

We conduct Monte Carlo simulations to evaluate the performance of PLM
and AIPW estimators under various data generating processes (DGPs).
Each DGP is characterized by:

\begin{itemize}
\item A binary covariate $X \sim \text{Bernoulli}(0.5)$
\item Two binary treatments $W_1, W_2$ with stratum-specific
propensity scores $p_j(X)$
\item Heterogeneous treatment effects $\tau_j(X)$ for each treatment
\end{itemize}

We consider five scenarios that vary in their degree of effect
heterogeneity and propensity score distributions:

\begin{enumerate}
\item \textbf{Extreme Heterogeneity}: Large differences in treatment
effects across strata with extreme propensity scores
\item \textbf{Constant Effects}: Homogeneous effects within treatments
but different across treatments
\item \textbf{Uncorrelated}: Moderate heterogeneity with balanced
propensity scores
\item \textbf{Selection on Gains}: Treatment probability correlated
with treatment effects
\item \textbf{Balanced}: Equal propensity scores across strata with
heterogeneous effects
\end{enumerate}

For each scenario, we simulate 1,000 datasets with 10,000 observations
each. We evaluate the estimators on three dimensions:
\begin{itemize}
\item Distribution of point estimates
\item Bias relative to true effects
\item Proportion of correct rankings between treatments
\end{itemize}

We report figures for each of these settings in appendix
\ref{appdx:simfigs}. We find that with the exception of the extreme
heterogeneity setting that expands upon the example in section 2 (fig
\ref{fig:extremehet}), the rankings produced by the PLM are largely
consistent with the AIPW estimator, and conform with the sufficient
conditions derived in the previous section.


\section{Conclusion}

This note highlights the importance of using appropriate methods for
estimating and ranking treatment effects in the presence of
heterogeneity. We show using an example that commonly used Partially
Linear Models can lead to biased estimates and incorrect rankings. We
then define a notion of ranking reversals and derive a decomposition relating the WATE and ATE, which gives rise to a necessary and sufficient
condition for ranking reversals in linear regression. Finally, we
propose interpretations for these conditions and recommend using
Augmented Inverse Probability Weighting estimator as a general
solution for ranking in the presence of substantial heterogeneity.

Our findings have important implications for decision-making in
various fields, including digital platforms and policy evaluation,
where accurate ranking of treatments is crucial. Future work could
explore the performance of these methods in more complex settings with
multiple treatments and high-dimensional covariates.

\renewcommand{\mkbibnamefamily}[1]{\textsc{#1}}
\printbibliography

@ARTICLE{Angrist1998-ok,
  title        = {Estimating the Labor Market Impact of Voluntary Military
                  Service Using Social Security Data on Military Applicants},
  author       = {Angrist, Joshua D},
  journaltitle = {Econometrica: journal of the Econometric Society},
  publisher    = {[Wiley, Econometric Society]},
  volume       = {66},
  issue        = {2},
  pages        = {249--288},
  date         = {1998},
  url          = {http://www.jstor.org/stable/2998558}
}

@ARTICLE{SloczynskiUnknown-kg,
  title   = {Interpreting {OLS} Estimands when Treatment Effects are
             Heterogeneous},
  date = {2022},
  journaltitle={Review of Economics and Statistics},
  author  = {Słoczyński, Tymon},
  urldate = {2024-05-01}
}

@ARTICLE{Chernozhukov2022-se,
  title        = {Automatic Debiased Machine Learning of Causal and Structural
                  Effects},
  author       = {Chernozhukov, Victor and Newey, Whitney K and Singh, Rahul},
  journaltitle = {Econometrica: journal of the Econometric Society},
  publisher    = {John Wiley \& Sons, Ltd},
  volume       = {90},
  issue        = {3},
  pages        = {967--1027},
  date         = {2022-05-01},
  url          = {https://doi.org/10.3982/ECTA18515},
  note         = {https://doi.org/10.3982/ECTA18515}
}

@ARTICLE{Imbens2004-ir,
  title        = {Nonparametric Estimation of Average Treatment Effects Under
                  Exogeneity: A Review},
  author       = {Imbens, Guido W},
  journaltitle = {The review of economics and statistics},
  publisher    = {MIT Press},
  volume       = {86},
  issue        = {1},
  pages        = {4--29},
  date         = {2004-02-01},
  url          = {https://doi.org/10.1162/003465304323023651},
  note         = {doi: 10.1162/003465304323023651}
}

@ARTICLE{Cattaneo2010-oc,
  title        = {Efficient semiparametric estimation of multi-valued treatment
                  effects under ignorability},
  author       = {Cattaneo, Matias D},
  journaltitle = {Journal of econometrics},
  volume       = {155},
  issue        = {2},
  pages        = {138--154},
  date         = {2010-04-01},
  url          = {https://www.sciencedirect.com/science/article/pii/S030440760900236X}
}

@ARTICLE{Aronow2016-nn,
  title        = {Does Regression Produce Representative Estimates of Causal
                  Effects?},
  author       = {Aronow, Peter M and Samii, Cyrus},
  journaltitle = {American journal of political science},
  volume       = {60},
  issue        = {1},
  pages        = {250--267},
  date         = {2016-01-28},
  url          = {https://onlinelibrary.wiley.com/doi/abs/10.1111/ajps.12185}
}

@INBOOK{Angrist1999-sp,
  title     = {Chapter 23 - Empirical Strategies in Labor Economics},
  author    = {Angrist, Joshua D and Krueger, Alan B},
  editor    = {Ashenfelter, Orley C and Card, David},
  booktitle = {Handbook of Labor Economics},
  publisher = {Elsevier},
  volume    = {3},
  pages     = {1277--1366},
  date      = {1999-01-01},
  url       = {http://www.sciencedirect.com/science/article/pii/S1573446399030047}
}

@ARTICLE{Chernozhukov2018-fl,
  title        = {Double/debiased machine learning for treatment and structural
                  parameters},
  author       = {Chernozhukov, Victor and Chetverikov, Denis and Demirer, Mert
                  and Duflo, Esther and Hansen, Christian and Newey, Whitney and
                  Robins, James},
  journaltitle = {The econometrics journal},
  volume       = {21},
  issue        = {1},
  pages        = {C1--C68},
  date         = {2018-02-16},
  url          = {http://doi.wiley.com/10.1111/ectj.12097}
}

\pagebreak

\appendix

\section{Proofs}

\subsection{Conditional Variance Weighting}\label{sec:p1}

We observe $(Y_i, W_i, \vee{X}_i)_{i=1}^N \in \R \times \{0, 1\} \times
\R^d$. We project the covariate vector $X_i$ into some basis $\Phi$, which approximates the flexible function $g(\vee{X})$.

\be
    \I Unconfoundedness: $Y_i^{0}, Y_i^{1} \indep W_i \mid X_i$
    \I Linearity of propensitye score $\E [W_i \mid X_i] = \phi_i'
    \psi$
\ee

Define $Z_i = (1 : W_i : \phi_i)$. We run the following regression
$$Y_i \sim Z_i = \alpha + \tau W_i + \Ubr{\phi_i' \zeta}{$g(x)$} +
\epsi_i
$$

By FWL, we can write the coefficient $\hat{\tau}$ as

\begin{align*}
    \hat{\tau} &= \frac{\sum_i \wt{W}_i Y_i}{\sum_i \wt{W}_i^2} &
   \wt{W}_i = W_i - \phi_i' \psi, \; \wh{\psi} = (\phi' \phi)^{-1}
   \phi' w \\ & = \frac{\sum_i \wt{W}_i \Bigpar{\yc_i + \tau_i
   W_i}}{\sum_i \wt{W}_i^2} =
   \frac{\sum_i \wt{W}_i \yc_i}{\sum_i \wt{W}_i^2} + \frac{\sum_i
   \wt{W}_i \tau_i W_i}{\sum_i \wt{W}_i^2} \\
   \\
   &= \Ubr{\frac{\sum_i \wt{W}_i \yc_i}{\sum_i \wt{W}_i^2}}{$\to 0$ by
   A1} +
   \Ubr{\frac{\sum_i \wt{W}_i \tau_i \phi_i' \psi}{\sum_i
   \wt{W}_i}}{$\to 0$ by orthogonality bw $\wt{W}_i$ and
   $\phi_i'\psi$} +
   \frac{\sum_i \wt{W}_i^2 \tau_i}{\sum_i \wt{W}_i^2} & \text{Expand
   out } W_i = \phi_i ' \psi + \wt{W}_i \\ &= \frac{\sum_i \wt{W}_i^2
   \tau_i}{\sum_i \wt{W}_i^2} =
   \frac{\sum_i (W_i - \phi_i'\psi)^2 \tau_i}{\sum_i (W_i -
   \phi_i'\psi)^2}
\end{align*}

\begin{proof}[Proof of necessity and sufficiency of \ref{eqn:cond} for rank reversal]\label{appdx:directproof}

We need it to be the case that \ref{eqn:cond}, combined with the definitional assumption that ATE$_j >$ ATE$_k$ ($\Exp[\tau_j(\vee{X})] > \Exp{\tau_k(\vee{X})}$) $\iff$ rank reversal WATE$_j < $ WATE$_k$. 

$\Leftarrow$ Using the decomposition \ref{eqn:watedecomp}, we note that LHS of \ref{eqn:cond} is equal to WATE$_j$, and the right hand side is WATE$_k$, so this is rank reversal by definition.

$\Rightarrow$ By the same token, since the LHS and RHS of \ref{eqn:cond} are the definition of WATE$_j$ and WATE$_k$ respectively by \ref{eqn:watedecomp}, this immediately implies the conclusion.

\end{proof}

\subsection{Interpretable Sufficient Conditions}\label{appdx:interpretablesuff}

\be
\I $\text{Cov}(\tau_j(\mathbf{X}), \gamma_j(\mathbf{X})) < -\delta$ for some $\delta > 0$
\I $\text{Cov}(\tau_k(\mathbf{X}), \gamma_k(\mathbf{X})) > \delta$
\I The difference in ATEs is smaller than the combined covariance in
effects: $\E[\tau_j(\vee{X})] - \E[\tau_k(\vee{X})]<2 \delta$
\ee

We proceed by showing that conditions (1)-(3) together imply rank
reversal as defined in Definition \ref{defn:rrev}. We need to show
that for treatment effect functions $\tau_j(\vee{X}),
\tau_k(\vee{X})$ that satisfy \ref{eqn:atejgk} and conditions (1-3),
\ref{eqn:wagejlk} holds.

Next, use this definition for $j$ and $k$ and plug in conditions (1)
and (2)
\begin{align*}
\mathbb{E}[\gamma_j(\mathbf{X})\tau_j(\mathbf{X})] &= \mathbb{E}[\tau_j(\mathbf{X})] + \text{Cov}(\tau_j(\mathbf{X}), \gamma_j(\mathbf{X})) \\
&< \mathbb{E}[\tau_j(\mathbf{X})] - \delta & \text{condition (1)} \\
\mathbb{E}[\gamma_k(\mathbf{X})\tau_k(\mathbf{X})] &=
\mathbb{E}[\tau_k(\mathbf{X})] + \text{Cov}(\tau_k(\mathbf{X}),
\gamma_k(\mathbf{X})) \\ &> \mathbb{E}[\tau_k(\mathbf{X})] + \delta &
\text{condition (2)}
\end{align*}

From condition (3): $\mathbb{E}[\tau_j(\mathbf{X})] -
\mathbb{E}[\tau_k(\mathbf{X})] < 2\delta$. Therefore:
\begin{align*}
\mathbb{E}[\gamma_j(\mathbf{X})\tau_j(\mathbf{X})] &< \mathbb{E}[\tau_j(\mathbf{X})] - \delta \\
&< \mathbb{E}[\tau_k(\mathbf{X})] + \delta & \text{plug in cond (3)} \\
&< \mathbb{E}[\gamma_k(\mathbf{X})\tau_k(\mathbf{X})]
& \square
\end{align*}

\pagebreak

\subsection{Simulation Study Results}\label{appdx:simfigs}

\begin{figure}
\includegraphics[width=\textwidth]{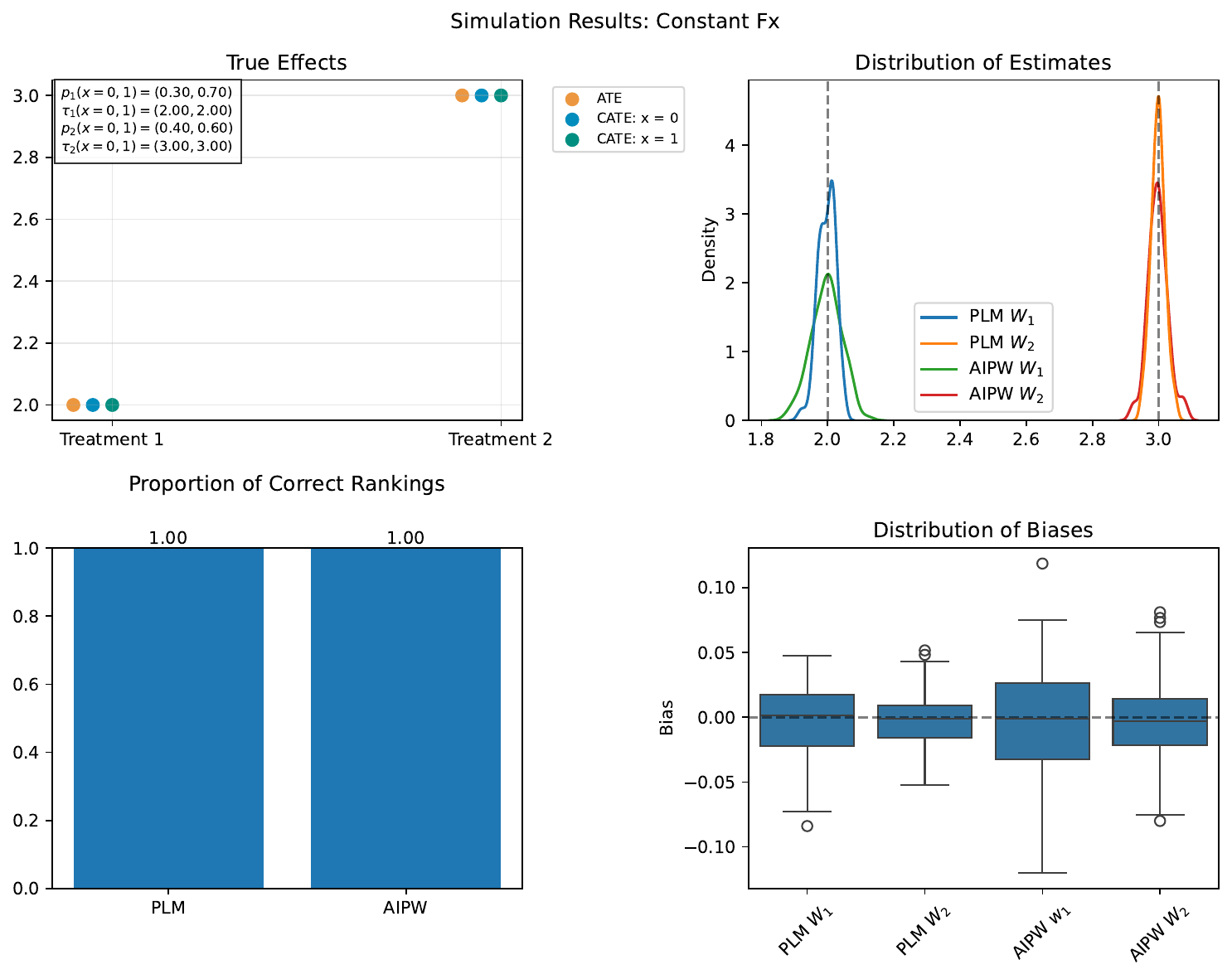}
\caption{Results for Constant Effects}
\label{fig:constfx}
\end{figure}

\begin{figure}
\includegraphics[width=\textwidth]{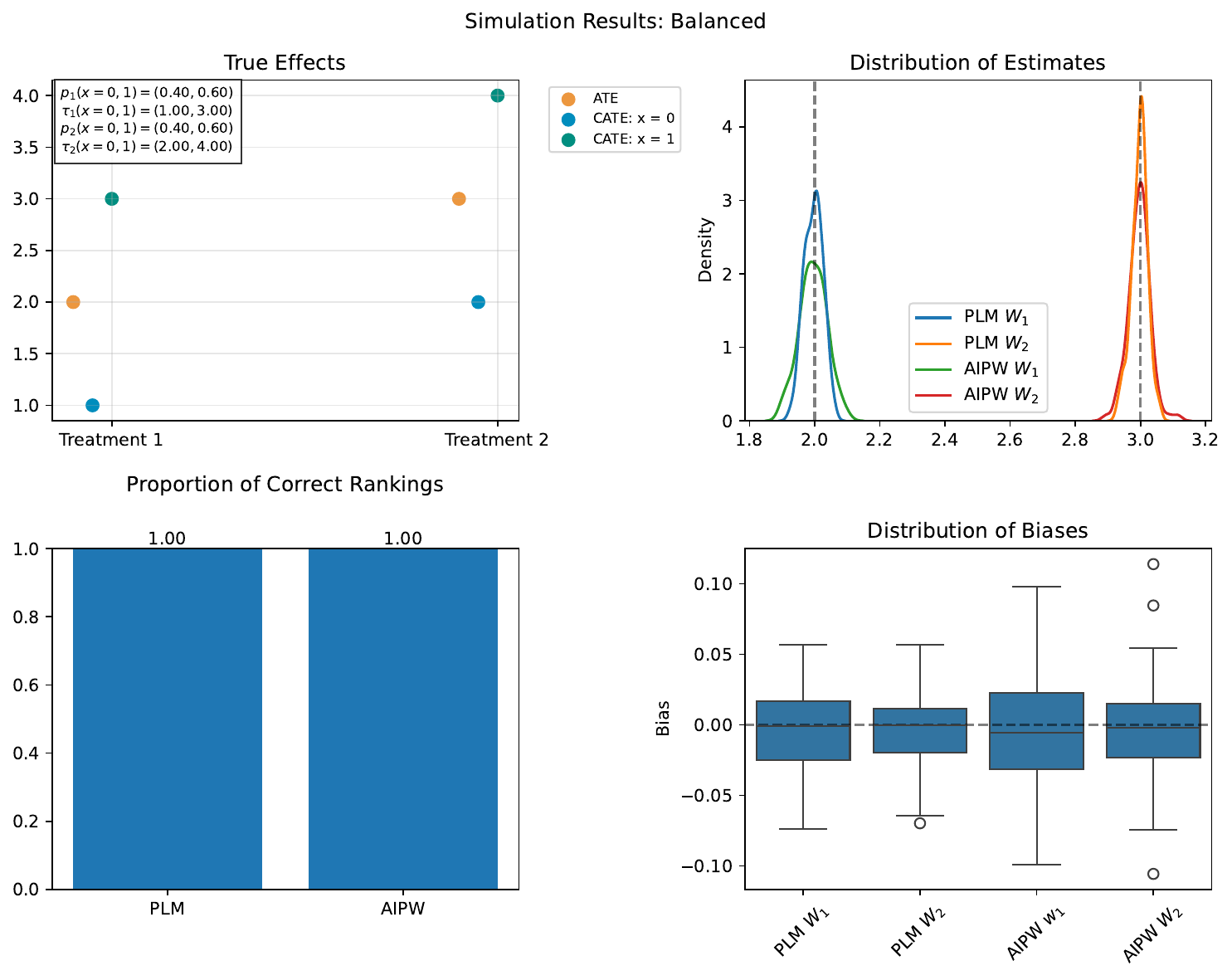}
\caption{results for balanced assignment}
\label{fig:bal}
\end{figure}

\begin{figure}
\includegraphics[width=\textwidth]{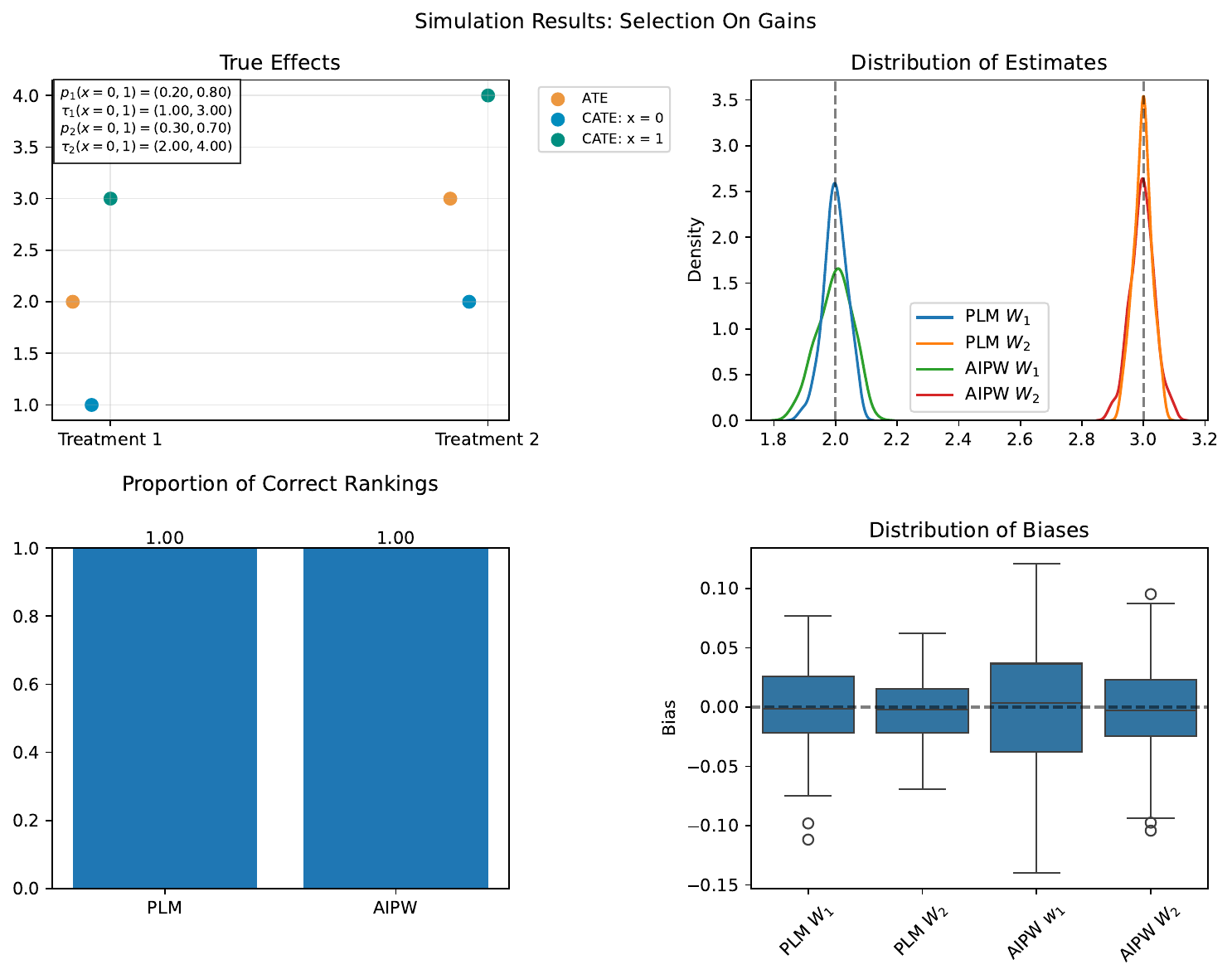}
\caption{Results for Selection on Gains}
\label{fig:selongains}
\end{figure}

\begin{figure}
\includegraphics[width=\textwidth]{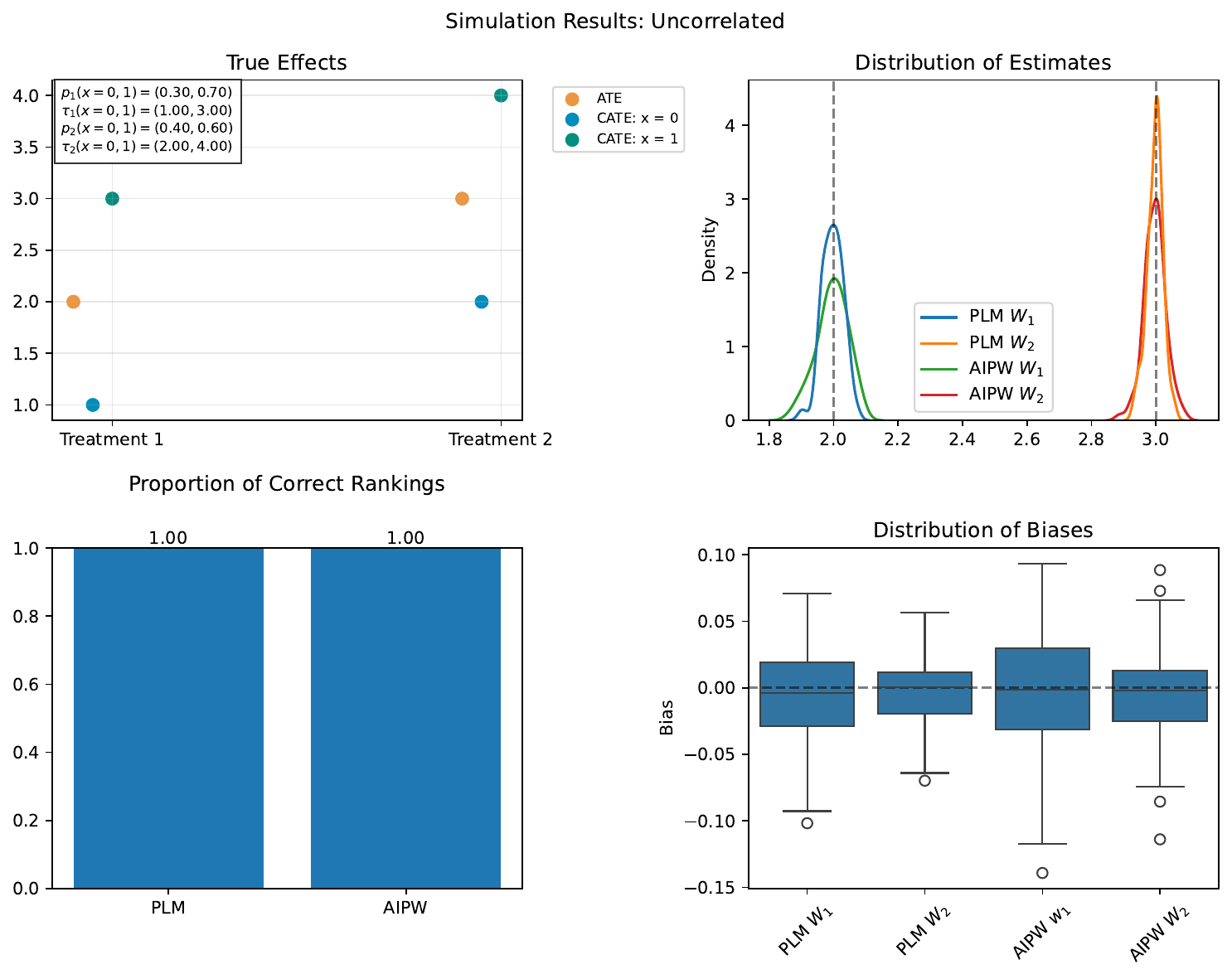}
\caption{Results for Uncorrelated propensity and treatment effects}
\label{fig:selongains}
\end{figure}

\begin{figure}
\includegraphics[width=\textwidth]{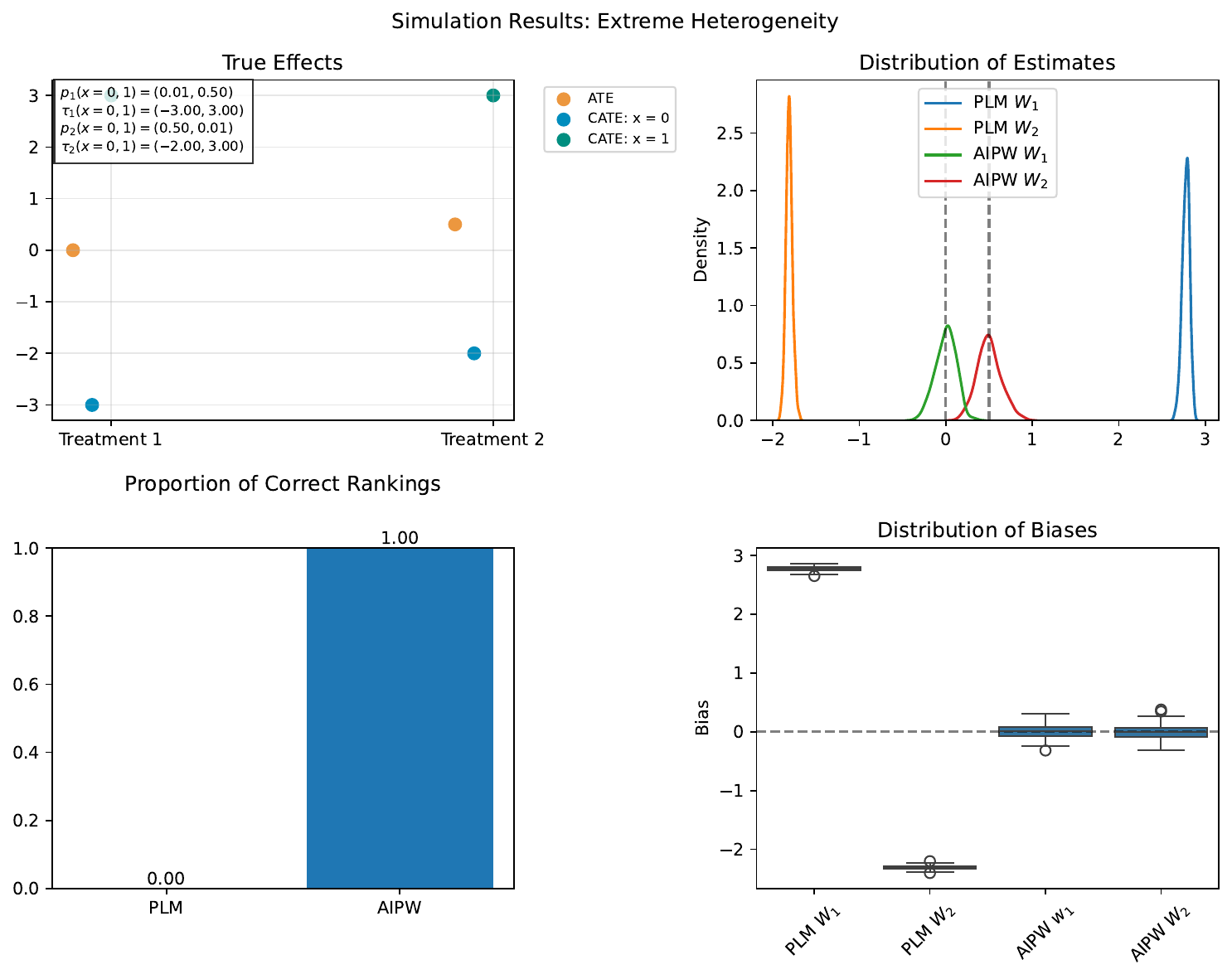}
\caption{Results for Extreme heterogeneity}
\label{fig:extremehet}
\end{figure}

\end{document}